\documentclass{article}

\usepackage{amsthm, amssymb, amsmath}
\usepackage{verbatim}
\usepackage{graphics}      
\usepackage{graphicx}      
\usepackage{setspace}
\usepackage{verbatim}
\usepackage{float}

\newtheorem{thm}{Theorem}

\newtheorem{lem}[thm]{Lemma}

\theoremstyle{plain}

\newcommand{\Z}{\mathbb{Z}}

\newcommand{\F}{\mathbb{F}}

\title{Unary Subset-Sum is in Logspace}
\author{Daniel M. Kane}

\begin{document}
\maketitle

\section{Introduction}

In this paper we consider the Unary-Subset-Sum problem which is defined as follows:
Given integers $m_1,\ldots,m_n$ and $B$ (written in unary), we define the subset sum problem to be that of determining whether or not there exists an $S\subseteq [n]$ so that $\sum_{i\in S} m_i = B$ (note that for this problem the $m_i$ are often assumed to be non-negative).  Let $C=|B| + \sum_{i=1}^n |x_i|+1$.  This problem can be solved using a standard dynamic program using space $O(C)$ and time $O(Cn)$. The dynamic program makes fundamental use of this large space and it is interesting to ask whether this requirement can be removed.  Unary Subset-Sum has been studied in small-space models of computation as early as 1980 in \cite{book}, where they showed that it was in $NL$.  Since then the problem was studied in \cite{comp}, where Cho and Huynh devised a complexity class between $L$ and $NL$ that contained Unary Subset-Sum as supporting evidence that it is not $NL$-complete.  This problem was listed again in \cite{overview} claiming it to be an open problem as to whether or not it is in $L$.  In 2010 it was recently shown in \cite{alternate} that this problem was in Logspace as a consequence of a much more general algorithm.  We provide a simple algorithm solving this problem in Logspace, which is also implementable in $\mathrm{TC}^0$.

\section{Our Algorithm}

The basic idea of our algorithm will be to make use of the generating function $\prod_{i=1}^n (1+x^{m_i}) = \sum_{S\subseteq[n]} x^{\sum_{i\in S} m_i}$ to compute the number of solutions to our problem modulo $p$ for a number of different primes $p$ (we show how to do this in Lemma \ref{mainLemma}).  Pseudocode for our algorithm is follows:

\begin{figure}[H]
\begin{tabbing}
\ \ \ \ \ \= \ \ \ \ \ \= \\
$c:=0$\\
$p: = \textrm{NextPrime}(C)$\\
$\textrm{While}(c \leq n)$\\
\> If $\sum_{x=1}^{p-1}x^{-B}\prod_{i=1}^n (1+x^{m_i}) \not\equiv 0 \pmod{p}$\\
\> \> Return True\\
\> $c:=c + \lfloor \log_2(p)\rfloor$\\
\> $p := \textrm{NextPrime}(p)$\\
End While\\
Return False
\end{tabbing}
\end{figure}

\subsection{Complexity}

There are several things that must be noted to show that this algorithm runs in logspace.  First, we claim that $p$ is never more than polynomial in size.  This is because standard facts about prime numbers imply that there are at least $n$ primes between $C$ and $\textrm{poly}(C,n)$, and each of these primes causes $c$ to increase by at least 1.  We also note that $\sum_{x=1}^{p-1}x^{-B}\prod_{i=1}^n (1+x^{m_i})$ can be computed modulo $p$ in Logspace.  This is because we can just keep track of the value of $x$ and the current running total (modulo $p$) along with the space necessary to compute the next term.  The product is computed again by keeping track of $i$ and the current running product (modulo $p$) and whatever is necessary to compute the next term.  The exponents are computed in the obvious way.  Finally primality testing of poly-sized numbers can be done by repeated trial divisions in Logspace, and hence the NextPrime function can also be computed in Logspace.

In fact, this function can also be computed in $\mathrm{TC}^0$. The function is clearly an OR over possible values of $p$. Each input requires computing a polynomial sized sum of polynomial sized products of sums of exponentials all modulo $p$. As all of these operations are known to be computable in $\mathrm{TC}^0$, the composition is as well.

\subsection{Correctness}

We now have to prove correctness of the algorithm.  Let $A$ be the number of subsets $S\subseteq [n]$ so that $\sum_{i\in S} m_i = B$.

\begin{lem}\label{mainLemma}
For $p$ a prime number, $p> C$.  Then
$$
\sum_{x=1}^{p-1}x^{-B}\prod_{i=1}^n (1+x^{m_i}) \equiv -A \pmod{p}.
$$
Where again $A$ is the number of subsets $S\subseteq [n]$ so that $\sum_{i\in S} m_i = B$.
\end{lem}
\begin{proof}
Note that $$x^{-B}\prod_{i=1}^n (1+x^{m_i})=\sum_{S\subseteq [n]} x^{\sum_{i\in S} m_i - B}.$$  The idea of our proof will be to interchange the order of summation and show that the terms for which $\sum_{i\in S} m_i \neq B$ cancel out.

Notice that each exponent in this sum has absolute value less than $p-1$.  Interchanging the sums on the right hand side, we find that
$$
\sum_{x=1}^{p-1}x^{-B}\prod_{i=1}^n (1+x^{m_i}) = \sum_{S\subseteq [n]} \sum_{x=1}^{p-1}x^{\sum_{i\in S} m_i - B}.
$$
We note that:
$$
\sum_{x=1}^{p-1} x^k \pmod{p} \equiv \begin{cases} -1 \ & \textrm{if} \ k \equiv  0  \pmod{p-1} \\ 0 \ & \textrm{else} \end{cases}.
$$
If $k$ is a multiple of $p-1$, then all terms in the sum are 1 modulo $p$ and the result follows.  Otherwise, we let $g$ be a primitive root mod $p$ and note that instead of summing over $x=1$ to $p-1$ we may sum over $x=g^\ell$ for $\ell=0$ to $p-2$.  Then
$$
\sum_{x=1}^{p-1} x^k  \equiv \sum_{\ell=0}^{p-2} g^{k\ell} \equiv \frac{1-g^{k(p-1)}}{1-g^k} \equiv \frac{1-1}{1-g^k} \equiv 0.
$$
Hence
$$
\sum_{x=1}^{p-1}x^{-B}\prod_{i=1}^n (1+x^{m_i}) = \sum_{S\subset [n]} \sum_{x=1}^{p-1}x^{\sum_{i\in S} m_i - B} \equiv \sum_{\substack{S\subseteq [n] \\ \sum_{i\in S} x_i \equiv B \pmod{p-1}}} -1.
$$
Since $p-1$ is larger than $C$, $\sum_{i\in S} x_i \equiv B \pmod{p-1}$ if and only if $\sum_{i\in S} x_i = B$.  Hence this sum contributes -1 for each such $S$ and so the final sum is $-A$.
\end{proof}

We are now ready to prove correctness. If $\sum_{x=1}^{p-1}x^{-B}\prod_{i=1}^n (1+x^{m_i}) \not\equiv 0 \pmod{p}$ for some $p>C$, then by our Lemma, this means that $A\not\equiv 0 \pmod{p}$.  In particular, this means that $A\neq 0$, and that therefore there is some such $S$.  Consider an integer $d$ which is equal to the product of the primes $p$ that have been checked so far.  Then $d$ is a product of distinct primes $p$ so that $-A \equiv \sum_{x=1}^{p-1}x^{-B}\prod_{i=1}^n (1+x^{m_i}) \equiv 0 \pmod{p}$.  Therefore $d|A$.  Furthermore it is the case that $d\geq 2^c$.  It is clear from the definition of $A$ that $0\leq A \leq 2^n$.  Therefore if $c>n$, $d>2^n$ and $d|A$, which implies that $A=0$, and that therefore there are no solutions.  Hence our algorithm always outputs correctly.

\section{Extensions}

There are some relatively simple extensions of this algorithm.  For one thing, our algorithm does more than tell us whether or not $A$ is equal to 0, but also tells us congruential information about $A$.  We can in fact obtain more refined congruential information than is apparent from our Lemma.  We can also use this along with the Chinese Remainder Theorem to compute a numerical approximation of $A$.  Finally a slight generalization of these techniques allows us to work with $m_i$ vector valued rather than integer-valued.

\subsection{Computing Congruences}

We show above how to compute $A$ modulo $p$ for $p$ a prime larger than $C$.  But in fact if $p$ is any prime and $k>1$ any integer, $A$ can be computed modulo $p^k$ in $O(\log((p+C)^k))$ space.

If $p>C$, then we have that
$$
A \equiv \frac{1}{p-1} \sum_{x=1}^{p-1}x^{-B}\prod_{i=1}^n (1+x^{m_i}) \pmod{p}.
$$
On the other hand if $p\leq C$, the above expression will only count the number of subsets that give the correct sum modulo $p-1$.  We can fix this by letting $q=p^\ell$ for some integer $\ell$ so that $q>C$.  Then for the same reasons that the above is true, it will be the case that
$$
A \equiv \frac{1}{q-1} \sum_{x\in \mathbb{F}_q^*}x^{-B}\prod_{i=1}^n (1+x^{m_i}) \pmod{p}.
$$
Where $\F_q$ is the finite field of order $q$.

If we have $k>1$ and $p>C$ we note that again for the same reasons
$$
A \equiv \frac{1}{p-1} \sum_{x\in \mu_{p-1}}x^{-B}\prod_{i=1}^n (1+x^{m_i}) \pmod{p^k}.
$$
Where $\mu_{p-1}$ is the set of $(p-1)^{st}$ roots of unity in $\Z/p^k$.  This computation can be performed without difficulty in $\Z/p^k$.  We again run into difficulty if $p<C$. This can be solved by performing the above computation in the Witt vectors of $\F_q$ modulo $p^k$ for $q>C$ some power of $p$, and taking the sum over $\mu_{q-1}$. This is at the cost of requiring $O(\log(q^r))$ space.

\subsection{Approximating the Number of Solutions}

It is also possible in Logspace to approximate the number of solutions, $A$, computing logarithmically many significant bits.  This can be done using the Chinese Remainder Theorem.  Suppose that $p_1,\ldots,p_k$ are distinct primes.  By the above we can compute $A$ modulo $p_i$ for each $i$.  Let $N=\prod_{i=1}^k p_i$, and $N_i=\frac{N}{p_i}$.  The Chinese Remainder Theorem tells us that
$$
A \equiv \sum_{i=1}^k N_i \left( A \pmod{p_i} \right) \left( N_i^{-1} \pmod{p_i}\right) \pmod{N}.
$$
Or in other words,
$$
\frac{A}{N} \equiv \sum_{i=1}^k \left( \frac{1}{p_i}\right)\left( A \pmod{p_i} \right) \left( N_i^{-1} \pmod{p_i}\right) \pmod{1}.
$$
Now we can compute $A$ modulo $p_i$ by the above.  We can also compute $N_i^{-1} \equiv \prod_{j\neq i} p_j^{-1} \pmod{p_i}$.  Hence we can compute each term in the sum to logarithmically many bits.  Hence in logspace we can compute
$$
\frac{A}{N} \pmod{1}
$$
to logarithmically many bits of precision.  If $2A>N>A$, this allows us to compute logarithmically many significant bits of $A$.  We can find such an $N$ by starting with an $N>2^n\geq A$ and repeatedly trying $N$ at least half as big as the previous $N$ until $N<2A$ (we can find our next $N$ by either removing the prime 2 from $N$ or replacing the smallest prime dividing $N$ by one at least half as big (which exists by Bertrand's postulate)).

It should also be noted that this ability to approximately count solutions in Logspace allows us to approximately uniformly sample from the space of solutions in Randomized Logspace.  This is done by deciding whether or not each element is in $S$ one-by-one and putting it in with probability nearly equal to the proportion of the remaining solutions that have that element in $S$.

It should also be noted that by performing the above computation modulo $m$ for any $m$, $A$ can be computed mod $m$ in $O(\log(m+C))$ space (though in a somewhat less elegant way than above).

\subsection{Vector-Valued Inputs}

We consider the slightly modified subset sum problem where now $m_i$ and $B$ lie in $\Z^k$, and again we wish to determine whether or not there exists and $S$ so that $\sum_{i\in S} m_i = B$.  If we let $C$ be one more than the sum of the absolute values of the coefficients of the $m_i$ plus the absolute values of the coefficients of $B$, a slight modification of our algorithm allows us to solve this problem in $O(k\log(C))$ space and $C^{O(k)}$ time (in particular if $k=O(1)$, this runs in $O(\log(C))$ space and $C^{O(1)}$ time).

There are two ways to do this.  One is simply to treat our vectors as base $C$-expansions of integers and reduce this to our previous algorithm.  Another technique involves a slight generalization of our Lemma.  In either case we let $m_i=(m_{i,1},\ldots,m_{i,k})$, $B=(B_1,\ldots,B_k)$.

For the first algorithm, we let $m_i' = \sum_{j=1}^k C^{j-1}m_{i,j}$ and $B'=\sum_{j=1}^k C^{j-1}B_{j}.$  We claim that for any $S\subseteq[n]$ that $\sum_{i\in S}m_i = B$ if and only if $\sum_{i\in S}m_i' = B'$, thus reducing this to an instance of our original problem.  The claim holds because
$$
\sum_{i\in S} m_i'-B' = \sum_{j=1}^k C^{j-1} \left(\sum_{i\in S} m_{i,j} - B_j \right) = \sum_{j=1}^k C^{j-1} e_j.
$$
Since the $e_j$ are all integers of absolute value less than $C$, this sum is 0 if and only if, each of the $e_j$ are 0.  Hence $\sum_{i\in S}m_i = B$ if and only if $\sum_{i\in S}m_i' = B'$.

Another way to do this is by generalizing our Lemma.  In particular it can be shown using similar techniques that if $A$ is the number of subsets $S$ that work, and if $p$ is a prime bigger than $C$ that
$$
-A \equiv \sum_{x_1,\ldots,x_k=1}^{p-1} \left(\prod_{i=1}^k x_i^{-B_i} \right)\left(\prod_{i=1}^n \left(1+\prod_{j=1}^k x_j^{m_{i,j}}\right) \right)\pmod{p}.
$$
Given this, there is a natural generalization of our algorithm.

It should also be noted that both of these techniques allow us to use the above-stated generalizations to our algorithm in the vector-valued context.

This generalization also allows us to solve some related problems, such as the Unary 0-1 Knapsack problem.  This problem is defined as follows: You are given a list of integer weights $w_1,\ldots,w_n$, a list of integer values, $v_1,\ldots,v_n$, and an integer bound $B$.  The objective is to find a subset $S\subseteq [n]$ so that$\sum_{i\in S} v_i$ is as large as possible subject to the restriction that $\sum_{i\in S}w_i \leq B$.  We do this by determining all possible pairs of $(\sum_{i\in S}w_i,\sum_{i\in S} v_i)$ by applying our algorithm to $m_i=(w_i,v_i)$ and $B=(w,v)$ for all $|w|\leq \sum_{i=1}^n |w_i|, |v|\leq \sum_{i=1}^n |v_i|$.  Of the pairs $(w,v)$ for which there is a solution, we keep track of the largest $v$ that corresponds to a $w\leq B$.  From this pair it is also not hard to use our algorithm to find a subset $S$ which achieves this bound.

\end{document}